\documentclass[a4paper,USenglish,cleveref, autoref]{lipics-v2019}

\newif{\ifdraft}\drafttrue
\draftfalse

\newif{\ifdrawFigures}\drawFigurestrue
\drawFiguresfalse

\newif{\iffull}\fulltrue

\usepackage{mathtools}
\usepackage[utf8]{inputenc}
\usepackage{amssymb,amsmath,enumerate}
\usepackage{fancyhdr}
\usepackage{stmaryrd}
\usepackage{xspace}
\usepackage{tikz}
\usetikzlibrary{positioning}
\usetikzlibrary{decorations.pathreplacing}
\usetikzlibrary{arrows.meta}
\usepackage{multicol}
 \usepackage{amsthm}

\makeatletter
\newcommand{\thickhline}{%
    \noalign {\ifnum 0=`}\fi \hrule height 1pt
    \futurelet \reserved@a \@xhline
}
\newcolumntype{"}{@{\hskip\tabcolsep\vrule width 1pt\hskip\tabcolsep}}
\makeatother

\newtheorem{conjecture}[theorem]{Conjecture}

\usepackage{amsfonts}
\usepackage{amssymb}
\usepackage{amsmath}

\usepackage{tabularx,hhline,rotating,xspace}
\usepackage{tikz}
\usetikzlibrary{calc}
\renewcommand{\setminus}{\mysetminus}

\newenvironment{aw}{\noindent\color{magenta} AW : }{}

\makeindex

\newcommand{\mysetminusD}{\raisebox{.8pt}{\hbox{\tikz{\draw[line width=0.6pt,line cap=round] (3.5pt,0pt) -- (0,5.2pt);}}}}
\newcommand{\mysetminusT}{\mysetminusD}
\newcommand{\mysetminusS}{\raisebox{.5pt}{\hbox{\tikz{\draw[line width=0.45pt,line cap=round] (2.2pt,0) -- (0,3.8pt);}}}}
\newcommand{\mysetminusSS}{\raisebox{.35pt}{\hbox{\tikz{\draw[line width=0.4pt,line cap=round] (1.5pt,0) -- (0,2.8pt);}}}}

\newcommand{\mysetminus}{\mathbin{\mathchoice{\mysetminusD}{\mysetminusT}{\mysetminusS}{\mysetminusSS}}}

\newcommand{\Nleq}{\mathrel{\unlhd}}
\newcommand{\Nle}{\mathrel{\lhd}}

\newcommand{\set}[2]{\left\{\, \mathinner{#1}\vphantom{#2}\: \left|\: \vphantom{#1}\mathinner{#2} \right.\,\right\}}
\newcommand{\oneset}[1]{\left\{\, \mathinner{#1} \,\right\}}
\newcommand{\interval}[2]{[ \mathinner{#1}..\mathinner{#2}] }

\newcommand{\abs}[1]{\left|\mathinner{#1}\right|}

\newcommand{\ceil}[1]{\left\lceil\mathinner{#1} \right\rceil}

\newcommand{\gen}[1]{\left< \mathinner{#1} \right>}
\newcommand{\genr}[2]{\left< \, \mathinner{#1}\; \middle|\;\mathinner{#2} \, \right>}

\newcommand{\Fit}{\operatorname{Fit}}
\newcommand{\FitL}{\operatorname{FitLen}}
\newcommand{\Set}{\mathrm{set}}

\newcommand{\mcomm}[3]{\left[{#1},\kern.1em_{#2}\,\kern.1em {#3} \right]}
\newcommand{\smcomm}[2]{\left[_{#1}\,\kern.1em {#2} \right]}

\newcommand{\fitcomm}[2]{\mcomm{#1}{M}{#2}}

\newcommand{\N}{\ensuremath{\mathbb{N}}}

\newcommand{\NP}{\ensuremath{\mathsf{NP}}\xspace} %
\newcommand{\coNP}{\ensuremath{\mathsf{coNP}}\xspace}

\newcommand{\Ac}[1]{\ensuremath{\mathsf{AC}^{#1}}\xspace}

\newcommand{\AC}{\ensuremath{\mathsf{AC}^0}\xspace}

\newcommand{\CC}{\ensuremath{\mathsf{CC}^0}\xspace}

\renewcommand{\P}{\ensuremath{\mathsf{P}}\xspace}

\renewcommand{\phi}{\varphi}
\newcommand{\eps}{\varepsilon}

\newcommand{\AND}{\ensuremath{\mathrm{AND}}\xspace}
\newcommand{\MOD}[1]{\ensuremath{\mathrm{MOD}_{#1}}\xspace}

\newcommand{\Oh}{\mathcal{O}}

\newcommand{\cA}{\mathcal{A}}

\newcommand{\cL}{\mathcal{L}}

\newcommand{\cH}{\mathcal{H}}
\newcommand{\cJ}{\mathcal{J}}

\newcommand{\cU}{\mathcal{U}}

\newcommand{\cX}{\mathcal{X}}
\newcommand{\cY}{\mathcal{Y}}
\newcommand{\cZ}{\mathcal{Z}}

\newcommand{\wt}[1]{\widetilde{#1}}

\newcommand{\SAT}{\textsc{3SAT}\xspace}

\newcommand{\EQNSAT}{\textsc{EQN-SAT}\xspace}
\newcommand{\EQNID}{\textsc{EQN-ID}\xspace}

\newcommand{\ProgSAT}{\textsc{ProgramSAT}\xspace}
\newcommand{\KColoring}[1]{\ensuremath{#1}\textsc{-Coloring}\xspace}

\newcommand{\sse}{\subseteq}

\newcommand\ie{i.e., }

\newcommand\eg{e.g.\xspace}

 \renewcommand{\theenumi}{\textup{(\roman{enumi})}}
 
 \renewcommand{\theenumiii}{\textup{\arabic{enumiii}.}}

\tikzset{
	ncbar angle/.initial=90,
	ncbar/.style={
		to path=(\tikztostart)
		-- ($(\tikztostart)!#1!\pgfkeysvalueof{/tikz/ncbar angle}:(\tikztotarget)$)
		-- ($(\tikztotarget)!($(\tikztostart)!#1!\pgfkeysvalueof{/tikz/ncbar angle}:(\tikztotarget)$)!\pgfkeysvalueof{/tikz/ncbar angle}:(\tikztostart)$)
		-- (\tikztotarget)
	},
	ncbar/.default=0.5cm,
}

\title{Hardness of equations over finite solvable groups under the exponential time hypothesis} 

\author{Armin Wei\ss}{%
Universität Stuttgart,
Institut für Formale Methoden der Informatik (FMI), Germany}{armin.weiss@fmi.uni-stuttgart.de}{https://orcid.org/0000-0002-7645-5867}{Funded by DFG project DI 435/7-1.}

\titlerunning{Equations over finite solvable groups}

\authorrunning{A. Wei\ss}

\Copyright{Armin Wei\ss}



\category{}

\iffull
\relatedversion{Conference version published in \cite{Weiss20}.}
\else
\relatedversion{Preprint available at \cite{Weiss20arxiv}.}
\fi

\supplement{}

\acknowledgements{I am grateful to Moses Ganardi for bringing my attention both to the \AND-weakness conjecture and to the exponential time hypothesis. I am also thankful to David A. Mix Barrington for an interesting email exchange concerning the \AND-weakness conjecture and the idea to include steps of the lower central series in \cref{prop:notANDweakrefined} to get a more refined upper bound.
	Furthermore, I am indebted to Caroline Mattes and Jan Philipp Wächter for many helpful discussions.
Finally, I want to thank the anonymous referees for their valuable comments.
}

\iffull
\nolinenumbers 
\hideLIPIcs  
\fi

\EventEditors{Artur Czumaj, Anuj Dawar, and Emanuela Merelli}
\EventNoEds{3}
\EventLongTitle{47th International Colloquium on Automata, Languages, and Programming (ICALP 2020)}
\EventShortTitle{ICALP 2020}
\EventAcronym{ICALP}
\EventYear{2020}
\EventDate{July 8--11, 2020}
\EventLocation{Saarbrücken, Germany (virtual conference)}
\EventLogo{}
\SeriesVolume{168}
\ArticleNo{102}

\begin{document}

\maketitle

	\begin{abstract}
Goldmann and Russell (2002) initiated the study of the complexity of the equation satisfiability problem in finite groups by showing that it is in \P for nilpotent groups  while it is \NP-complete for non-solvable groups. Since then, several results have appeared showing that the problem can be solved in polynomial time in certain solvable groups of Fitting length two. In this work, we present the first lower bounds for the equation satisfiability problem in finite solvable groups: under the assumption of the exponential time hypothesis, we show that it cannot be in \P for any group of Fitting length at least four and for certain groups of Fitting length three. Moreover, the same hardness result applies to the equation identity problem.
	\end{abstract}

\iffull\else
\setcounter{page}{0} 
\newpage
\fi


  \section{Introduction}
  
  The study of equations over algebraic structures has a long history in mathematics. Some of the first explicit decidability results in group theory are due to Makanin \cite{mak77}, who showed that equations over free groups are decidable. Subsequently several other decidability and undecidability results as well as complexity results on equations over infinite groups emerged (see \cite{DiekertE17icalpshort,GarretaMO20,LohSen06,Romankov79} for a random selection). For a fixed group $G$, the equation satisfiability problem $\EQNSAT$ is as follows: given an expression $\alpha \in (G\cup  \cX \cup \cX^{-1})^*$ where $\cX$ is some set of variables, the question is whether there exists some assignment $\sigma: \cX \to G$ such that $\sigma(\alpha)=1$ (here $\sigma$ is extended to expressions in the natural way~-- $\cX^{-1}$ is a disjoint copy of $\cX$ representing the inverses of $\cX$). Likewise \EQNID is the problem, given an expression, decide whether it evaluates to 1 under \emph{all} assignments.
  
   Henceforth, all groups we consider are finite. In this case, equation satisfiability and related questions are clearly decidable by an exhaustive search. 
  Still the complexity is an interesting topic of research: its study has been initiated by Goldmann and Russell \cite{GoldmannR02}, who showed that satisfiability of systems of equations can be decided in $\P$ if and only if the group is abelian (assuming $\P \neq \NP$)~-- otherwise, the problem is \NP-complete. They also obtained some results for single equations: $\EQNSAT$ is \NP-complete for non-solvable groups, while for nilpotent groups it is in \P. 
  This left the case of solvable but non-nilpotent groups open. 
  Indeed, Burris and Lawrence raised the question whether $\EQNID(G) \in \P$ for all finite solvable groups $G$ \cite[Problem 1]{BurrisL04}. Moreover, Horváth \cite{Horvath11} conjectured a positive answer.

 \subparagraph*{Contribution.} In this work we give a negative answer to this question assuming the exponential time hypothesis by showing the following result: 
    \theoremstyle{plain}
  \newtheorem{corollaryA}{Corollary}
  \renewcommand*{\thecorollaryA}{\Alph{corollaryA}}
  \begin{corollaryA}\label{cor:mainIntro}
  	Let $G$ be finite solvable group and assume that either
  	\begin{itemize}
  		\item the Fitting length of $G$ is at least four, or
  		\item the Fitting length of $G$	is three and there is no Fitting-length-two normal subgroup whose index is a power of two.
  	\end{itemize}
  	Then $\EQNSAT(G)$ and $\EQNID(G)$ are not in \P under the exponential time hypothesis.
  \end{corollaryA}
To the best of our knowledge, this constitutes the first hardness results for $\EQNSAT(G)$ and $\EQNID(G)$ if $G$ is solvable.%
\footnote{Recently (a preprint appeared only days after the submission of this paper), in \cite{IdziakKK20} Idziak, Kawa\l{}ek, and Krzaczkowski succeeded to show that $\EQNSAT(S_4)$ is not in \P under the exponential time hypothesis ($S_4$ denotes the symmetric group over four elements). Moreover, they proved similar results as in this work for the case of algebras from congruence modular varieties. This complements our main result Corollary~\ref{cor:mainIntro}. Indeed, a joint paper proving a quasipolynomial lower bound on \EQNSAT and \EQNID for \emph{all} finite groups of Fitting length three can be found in \cite{IdziakKKW20arxiv}.
}
The Fitting length of a group $G$ is the minimal $d$ such that there is a sequence $1 = G_0 \Nleq \cdots \Nleq G_d = G$ with all quotients $G_{i+1} /G_i$ nilpotent.
  
  Moreover, we show that if $S$ is a semigroup with a group divisor (\ie a group which is a quotient of a subsemigroup of $S$) meeting the requirements of Corollary~\ref{cor:mainIntro}, $\EQNSAT(S)$  (here the input consists of two expressions)  is also not in \P under the exponential time hypothesis.  
   Finally, using the same ideas as for our main result, we derive an upper bound of $2^{\Oh(n^{1/(d-1)})}$ for the length of the shortest $G$-program (definition see below) for the $n$-input \AND function in a finite solvable group of Fitting length $d \geq 2$. Notice that a corresponding $2^{n^{\Omega(1)}}$ lower bound would imply that $\EQNSAT(G)$ and $\EQNID(G)$ can be solved in quasipolynomial time for finite solvable groups $G$.

  \subparagraph*{General approach.}
  The complexity of \EQNSAT is closely related to the complexity of the satisfiability problem for $G$-programs (denoted by \ProgSAT~-- for a definition see \cref{sec:programs}). Indeed, \cite{BarringtonMMTT00} gives a reduction from \EQNSAT to \ProgSAT (be aware that, while the problems \EQNSAT and \ProgSAT are well-defined for finitely generated infinite groups, in general, such a reduction exists only in the case of finite groups). Moreover, also \ProgSAT is in \P for nilpotent groups and \NP-complete for non-solvable groups \cite{BarringtonST90}.

  In order to show hardness of these problems, one usually reduces some \NP-complete problem like \SAT or \KColoring{C} to them. Typically, this requires to encode big logical conjunctions into the group $G$.
 Therefore, the complexity of these problems is linked to the length of the shortest $G$-program for the \AND function. 
 	Indeed, \cite[Theorem 4]{BarringtonMMTT00} shows that, if the \AND function can be computed by a \P-uniform family of $G$-programs of polynomial length, then $\ProgSAT(G\wr C_k)$ for $k\geq 4$ is \NP-complete (here $C_k$ denotes the cyclic group of order $k$; \P-uniform means that the $n$-input $G$-program can be computed in time polynomial in $n$).
  Thus, if there exists a solvable group with efficiently computable polynomial length $G$-programs for the \AND function, then there is a solvable group with an \NP-complete \ProgSAT problem.

  It is well-known that $G$-programs describe the circuit complexity class \CC \cite{McKenziePT91} with the depth of the circuit relating to the Fitting length of the group. One can make a depth size trade-off for the \AND function using a divide-and-conquer approach: Assume there is a circuit of depth two and size $2^n$ for the $n$-input \AND (which is the case by \cite{Barrington85}). Since the $n$-input \AND can be decomposed as $\sqrt{n}$-input \AND of $\sqrt{n}$ many $\sqrt{n}$-input \AND{}s, we obtain a \CC circuit of depth $4$ and size roughly $2^{\sqrt{n}}$.

  This observation plays a crucial role for our results: it allows us to reduce an $m$-edge \KColoring{C} instance to an equation of size roughly $2^{\sqrt{m}}$. We compare this to the exponential time hypothesis (ETH), which conjectures that $n$-variable \SAT cannot be solved in time $2^{o(n)}$. ETH implies that \KColoring{C} cannot be solved in time $2^{o(m)}$, which gives us a quasipolynomial lower bound on \EQNSAT and \EQNID. 
  Notice that in the literature there are several other quasipolynomial lower bounds building on the exponential time hypothesis~-- see 
  \cite{AaronsonIM14,BravermanKRW17,BravermanKW15} for some examples.

  \subparagraph*{Outline.}

  In \cref{sec:prelims}, we fix our notation and state some basic results on inducible and atomically universally definable subgroups. Some of these observations are well-known, while others, to the best of our knowledge, have not been stated explicitly.
  \Cref{sec:programs} gives a little excursion to the complexity of the \AND-function in terms of $G$-programs over finite solvable groups deriving an upper bound $2^{\Oh(n^{1/(d-1)})}$ if $d \geq 2$ is the Fitting length of $G$.

  \Cref{sec:reduction} and \cref{sec:consequences} are the main part of our paper:  we reduce the \KColoring{C} problem to \EQNSAT and \EQNID. For the reduction, we need some special requirements on the group $G$. In \cref{sec:consequences} we show that actually the requirements of Corollary~\ref{cor:mainIntro} are enough using the concept of inducible and atomically universally definable subgroups.
  Finally, in \cref{cor:semigroup} we examine consequences to \EQNSAT in semigroups.

  \subparagraph*{Related work on equations.}
 Since the work of Goldman and Russell \cite{GoldmannR02} and Barrington et.\,al.\ \cite{BarringtonMMTT00}, a long list of literature has appeared investigating \EQNID and \EQNSAT in groups and other algebraic structures. 
  In \cite{BurrisL04} it is shown that \EQNID is in \P for nilpotent groups as well as for dihedral groups $D_k$ where $k$ is odd. 
  Horváth resp.\ Horváth and Szabó \cite{Horvath15,HorvathS06} extended these results by showing the following among other results: $\EQNSAT(G) $ is in $ \P$ for $G = C_n \rtimes B$ with $B$ abelian, $n=p^k$ or $n=2p^k$ for some prime $p$ and \EQNID is in \P for semidirect products $G = C_{n_1} \rtimes (C_{n_2} \rtimes \cdots \rtimes ( C_{n_k} \rtimes(A\rtimes B)))$ with $A,B$ abelian (be aware that such a group is two-step solvable).
  Furthermore, in \cite{Foldvari17} it is proved that $\EQNSAT(G) \in \P$ for so-called semi-pattern groups. 
 Finally, in \cite{FoldvariH19} Földvári and Horváth established that \EQNSAT is in \P for the semidirect product of a $p$-group and an abelian group and that \EQNID is in \P for the semidirect product of a nilpotent group with an abelian group.
   Notice that all these groups have in common that their Fitting length is at most two.

In  \cite{HorvathS11,HorvathS12} the \EQNSAT and \EQNID problems for generalized terms are introduced. Here a generalized term means an expression which may also use commutators or even more complicated terms inside the input expression. Using commutators is a more succinct representation, which allows for showing that  \EQNSAT is \NP-complete and \EQNID is \coNP-complete in the alternating group $A_4$ \cite{HorvathS12}. 
In \cite{Kompatscher19} this result is extended by showing that, with commutators and the generalized term $w(x,y_1,y_2,y_3) = x^8[x,y_1,y_2,y_3]$, \EQNSAT is \NP-complete and \EQNID is \coNP-complete for all non-nilpotent groups.

There is also extensive literature on equations in other algebraic structures~-- for instance, 
\cite{AlmeidaVG09,BarringtonMMTT00,JacksonM06,Kisielewicz04,KlimaTT07,Klima09,Seif05,SeifS06,SzaboV04} in semigroups. 
 We only mention two of them explicitly:
\cite{Kisielewicz04} showed that identity checking (\EQNID without constants in the input) in semigroups is \coNP complete. 
Moreover, among other results, \cite{AlmeidaVG09} reduces the identity checking problem in the direct product of maximal subgroups to identity checking in some semigroup.

  \section{Preliminaries}\label{sec:prelims}
  
  The set of words over some alphabet $\Sigma$ is denoted by $\Sigma^*$. The length of a word $w \in \Sigma^*$ is denoted by $\abs{w}$. We denote the interval of integers $\oneset{i, \dots, j}$ by $\interval{i}{j}$.

\subparagraph*{Complexity.}
We use standard notation from complexity theory. In several cases we use the notion of \AC many-one reductions (denoted by $\leq_{\mathrm{m}}^{\AC}$) meaning that the reducing function can be computed in \AC (\ie by a polynomial-size, constant-depth Boolean circuit). The reader unfamiliar with this terminology may think about logspace or polynomial time reductions. Also be aware that in order to obtain \AC many-one reductions in most cases we need the presence of a letter representing the group identity for padding reasons.

\subparagraph*{Exponential time hypothesis.}  
The exponential time hypothesis (ETH) is the conjecture that there is some $\delta > 0$ such that every algorithm for $\SAT$ needs time $\Omega(2^{\delta n})$ in the worst case where $n$ is the number of variables of the given $\SAT$ instance. By the sparsification lemma \cite[Thm.~1]{ImpagliazzoPZ01} this is equivalent to the existence of some $\epsilon > 0$ such that every algorithm for $\SAT$ needs time $\Omega(2^{\epsilon (m+n)})$ in the worst case where $m$ is the number of clauses of the given $\SAT$ instance (see also \cite[Thm.~14.4]{CyganFKLMPPS15}). In particular, under ETH there is no algorithm for \SAT running in time $2^{o(n+m)}$.

\subparagraph*{\KColoring{C}.} A $C$-coloring for $C\in \N$ of a graph $\Gamma = (V,E)$ is a map $\chi:V \to \interval{1}{C}$. A coloring $\chi$ is called \emph{valid} if $\chi(u) \neq \chi(v)$ whenever $\oneset{u,v }\in E$.
The problem \KColoring{C} is as follows: given an undirected graph $\Gamma = (V,E)$, the question is whether there is a valid $C$-coloring of $\Gamma$. The \KColoring{C} problem is one of the classical \NP-complete problems for $C\geq 3$. Moreover, by \cite[Thm.~14.6]{CyganFKLMPPS15}, \KColoring{3} cannot be solved in time $2^{o(\abs{V} + \abs{E})}$ unless ETH fails. Since  \KColoring{3} can be reduced to  \KColoring{C} for fixed $C \geq 3$ by introducing only a linear number of additional edges and a constant number of vertices, it follows for every $C\geq 3$ that also \KColoring{C} cannot be solved in time $2^{o(\abs{V} + \abs{E})}$ unless ETH fails.

\subparagraph*{Commutators and Fitting series.}
Throughout, we only consider finite groups $G$. We use notation similar to \cite{Robinson96book}. 
  We write $[x,y] = x^{-1}y^{-1}xy$ for the commutator and $x^y = y^{-1}xy$ for the conjugation. Moreover, we write $[x_1, \dots, x_n] = [[x_1, \dots, x_{n-1}],x_n]$ for $n\geq 3$.

As usual for subsets $X,Y \sse G$, we write $\gen{X}$ for the subgroup generated by $X$ and we define $[X,Y] = \genr{[x,y]}{x \in X, y \in Y}$ and $[X_1,\dots, X_k] = [[X_1,\dots,X_{k-1}], X_k]$ for $X_1,\dots, X_k \sse G$. In contrast, we write
$[X,Y]_{\Set} = \set{[x,y]}{x \in X, y \in Y}$  (thus, $[X,Y] = \gen{[X,Y]_{\Set}}$) and $[X_1,\dots, X_k]_{\Set} = [[X_1,\dots,X_{k-1}]_{\Set}, X_k]_{\Set}$.

 Finally, we denote the set $\set{g^x}{x\in X}$ with $g^X$  (be aware that here we differ from \cite{Robinson96book}) and define $X^Y = \set{x^y}{x\in X ,y \in Y}$.

\begin{lemma}\label{lem:setcommutator}
	If $X_i^G=X_i \sse G$ for $i= 1, \dots, k$, then
	\[ [\gen{X_1},\dots,\gen{ X_k}] = \gen{[X_1,\dots, X_k]_{\Set}}.\]
\end{lemma}
\begin{proof}
By \cite[5.1.7]{Robinson96book}, we have $[\gen{X}, \gen{Y}] = \gen{[X, Y]^{\gen{X}\gen{Y}}}$ for arbitrary $X,Y \sse G$. Thus, if $X=X^G$ and $Y = Y^G$, we have $[\gen{X}, \gen{Y}] = [X, Y]$. 
We use this to show the lemma by induction:
\begin{align*}
[\gen{X_1},\dots,\gen{ X_k}] &= \big[[\gen{X_1},\dots,\gen{X_{k-1}}], \gen{ X_k}\big]\\
&=\big[ \gen{[X_1,\dots, X_{k-1}]_{\Set}}, \gen{ X_k}\big]\tag{by induction}\\
&=\big[ [X_1,\dots, X_{k-1}]_{\Set},  X_k\big] \tag{by \cite[5.1.7]{Robinson96book}}\\
&=\gen{[X_1,\dots, X_k]_{\Set}}\qedhere
\end{align*}
\end{proof}

For $x,y \in G$, we write $\mcomm{x}{k}{y} = [x,\underbrace{y,\dots,y}_{k\text{ times}}]$ and likewise for $X,Y \sse G$, we write $\mcomm{X}{k}{Y} = [X,\underbrace{Y,\dots,Y}_{k\text{ times}}]$ and $\smcomm{k}{Y} = [\underbrace{Y,\dots,Y}_{k\text{ times}}]$ and analogously  $\mcomm{X}{k}{Y}_{\Set} $ and $\smcomm{k}{Y}_{\Set}$.

Since $G$ is finite, there is some $M = M(G) \in \N$ such that $\mcomm{X}{M}{Y} = \mcomm{X}{i}{Y}$ for all $i \geq M$ and all $X,Y\sse G$ with $X^G=X$ and $Y^G=Y$ (notice that $\mcomm{X}{i}{Y} \leq \mcomm{X}{j}{Y}$ for $j\leq i$ due to the normality of $[X,Y]$). It is clear that $M =\abs{G}$ is large enough, but typically much smaller values suffice.

\begin{lemma}\label{lem:addGincomm}
	For all $X, Y \sse G$ with $X^G = X$ and $Y = Y^G$ we have $\fitcomm{X}{Y} = \fitcomm{[X,G]}{Y}$.
\end{lemma}
\begin{proof}
	We have $[X,G] \leq \gen{X}$ because $[x,g] = x^{-1} x^g \in X$. Thus, the inclusion right to left follows. The other inclusion is because  $\fitcomm{X}{Y} = \mcomm{X}{M+1}{Y} \leq [X,G, \kern.1em_{M}\kern.1em Y] = \fitcomm{[X,G]}{Y}$.
\end{proof}

The $k$-th term of the lower central series is $\gamma_kG = \mcomm{G}{k}{G}$. The \emph{nilpotent residual} of $G$ is defined as $\gamma_\infty G = \gamma_MG$ where $M$ is as above (\ie $\gamma_\infty G = \gamma_iG$ for every $i \geq M$). Recall that a finite group $G$ is nilpotent if and only if     $\gamma_\infty G = 1$.

The \emph{Fitting} subgroup $\Fit(G)$ is the union of all nilpotent normal subgroups. Let $G$ be a finite solvable group. It is well-known that $\Fit(G)$ itself is a nilpotent normal subgroup (see \eg \cite[Satz 4.2]{Huppert67}). The \emph{upper Fitting series} 
\[1 = \cU_0G \Nle \cU_1G \Nle \cdots \Nle \cU_k G = G \]
is defined by $\cU_{i+1}G/\cU_iG = \Fit(G/\cU_iG)$.
The \emph{lower Fitting series}
\[1 = \cL_dG \Nle \cdots \Nle \cL_1 G  \Nle \cL_0 G = G \]
is defined by $\cL_{i+1}G = \gamma_\infty(\cL_{i}G)$. We have $d=k$  (see \eg \cite[Satz 4.6]{Huppert67}) and this number is called the \emph{Fitting length} $\FitL(G)$ (sometimes also referred to as \emph{nilpotent length}). 
The following fact can be derived by a straightforward induction from the characterization of $\Fit(G)$ as largest nilpotent normal subgroup (for a proof see \eg \cite{mathoverflowW20}):
\begin{lemma}\label{lem:Fitting}
	Let $H \Nleq G$ be a normal subgroup. Then for all $i$, we have $\cU_i H = \cU_iG \cap H$. In particular,
	\begin{enumerate}
		\item if $\FitL(H) = i$, then $H \leq \cU_iG$,
		\item if $g \in \cU_{i}G \setminus \cU_{i-1}G$, then $\FitL(\gen{g^G}) = i$.
	\end{enumerate}
\end{lemma}

\subparagraph*{Equations in groups.}

An \emph{expression} (also called a \emph{polynomial} in \cite{SeifS06,HorvathS06,Kompatscher19}) over a group $G$ is a word $\alpha$ over the alphabet $G \cup \cX \cup \cX^{-1}$ where $\cX$ is a set of variables. Here $\cX^{-1} $ denotes a formal set of inverses of the variables. Since we are dealing with finite groups only, a variable $X^{-1}\in \cX^{-1}$ for $X\in \cX$ can be considered as an abbreviation for $X^{\abs{G}-1}$.
Sometimes we write $\alpha(X_1, \dots, X_n)$ for an expression $\alpha$ to indicate that the variables occurring in $\alpha$ are from the set $\oneset{X_1, \dots, X_n}$. Moreover, if $\beta_1, \dots, \beta_n$ are other expressions, we write $\alpha(\beta_1, \dots, \beta_n)$ for the expression obtained by substituting each occurrence of a variable $X_i$ by the expression $\beta_i$.

An assignment for an expression $\alpha$ is a mapping $\sigma:\cX \to G$~-- here $\sigma$ is canonically extended by $\sigma( X^{-1}) = \sigma(X)^{-1}$ and $\sigma(g) = g$ for $g \in G$. An assignment $\sigma$ is \emph{satisfying} if $\sigma(\alpha)=1$ in $G$.
The problems $\EQNSAT(G)$ and $\EQNID(G)$ are as follows: for both of them the input is an expression $\alpha$.
For $\EQNSAT(G)$ the question is whether there \emph{exists} a satisfying assignment, for $\EQNID(G)$ the question is whether \emph{all} assignments are satisfying.

Notice that in the literature \EQNSAT is also denoted by POL-SAT \cite{SeifS06,HorvathS06} or $\mathsf{Eq}$ \cite{Kompatscher19},
while \EQNID is also referred to as POL-EQ (\eg in \cite{SeifS06,HorvathS06,KlimaTT07}) or $\mathsf{Id}$ \cite{Kompatscher19}.

If $\cX = \cY \cup \cZ$ with $\cY \cap \cZ = \emptyset$ and we are given assignments $\sigma_1: \cY \to G$ and $\sigma_2: \cZ \to G$, we obtain a new assignment $\sigma_1\cup \sigma_2$ defined by $(\sigma_1\cup \sigma_2) (X)  = \sigma_1(X)$ if $X \in \cY$ and  $(\sigma_1\cup \sigma_2) (X)  = \sigma_2(X)$ if $X \in \cZ$. We write $[X \mapsto g]$ for the assignment $\oneset{X} \to G$ mapping $X $ to $g$.

\subparagraph*{Inducible subgroups.} 
According to \cite{GoldmannR02}, we call a subset $S \sse G$ \emph{inducible} if there is some expression $\alpha \in (G \cup \cX \cup \cX^{-1})^*$ such that $S = \set{\sigma(\alpha)}{\sigma\colon \cX \to G}$. In this case we say that $\alpha$ \emph{induces} $S$. Notice that in a finite group every verbal subgroup is inducible. (A subgroup is called \emph{verbal} if it is generated by a set of the form $\set{\sigma(\alpha)}{\sigma\colon \cX \to G, \alpha \in \cA}$ where $\cA \sse (\cX \cup \cX^{-1})^*$ is a \emph{finite} set of expressions without constants.)
This shows the first three points of the following lemma (for $\gamma_1G$, see  also \cite[Lemma 5] {GoldmannR02}): 
\begin{lemma}\label{lem:inducible}
	Let $G$ be a finite group. Then
	\begin{enumerate}
		\item for every $k \in \N$, the subgroup generated by all $k$-th powers is inducible,
		\item every element $\gamma_kG$ of the lower central series is inducible, 
		\item every element $\cL_kG$ of the lower Fitting series is inducible,
		\item if $K\leq H \leq G$ and $K$ is inducible in $H$ and $H$ inducible in $G$, then $K$ is also inducible in $G$,
		\item if $H\leq G$ with $H = [G,H]$, then $H$ is inducible.
	\end{enumerate}	
\end{lemma}
The fourth point follows simply by ``plugging in'' an expression for $H$ inside an expression for $K$.
The last point follows from the proof of \cite[Lemma 9 ]{Kompatscher19}.

The notion of inducible subgroup turns out to be very useful for proving lower bounds on the complexity. Indeed, the following facts are straightforward: 
\begin{lemma}[{\,\!\cite[Lemma 8]{GoldmannR02}, \cite[Lemma 9,  10]{HorvathS11}}]\label{lem:inducibleEQN}
	Let $H \leq G$ be an inducible subgroup. Then
	\begin{itemize}
		\item $\EQNSAT(H) \leq_{\mathrm{m}}^{\Ac0} \EQNSAT(G)$, and 
		\item $\EQNID(H) \leq_{\mathrm{m}}^{\Ac0} \EQNID(G)$.
		\item If, moreover, $H$ is normal in $G$, then $\EQNSAT(G/H) \leq_{\mathrm{m}}^{\Ac0} \EQNSAT(G)$. 
	\end{itemize}	
\end{lemma}

Let us briefly sketch the ideas to see this lemma: Fix an expression $\beta$ inducing $H$. For first and second reduction, replace every occurring variable of a given equation by a copy of $\beta$ with disjoint variables. The third reduction simply appends $\beta$ to an input equation.

\subparagraph*{Atomically universally definable subgroups.} 
The situation for reducing $\EQNID(G/H)$ to $ \EQNID(G)$ is slightly more complicated. For this we need a new definition:
We call a subset $S \sse G$ \emph{atomically universally definable} if there is some expression $\alpha \in (G \cup \cX \cup \cX^{-1})^*$ where $\cX = \oneset{X} \cup \oneset{Y_1, Y_2,\dots }$ such that \[S = \set{g \in G}{(\sigma \cup [X \mapsto g])(\alpha) = 1 \text{ for all } \sigma\colon \oneset{Y_1, Y_2,\dots } \to G}.\] In this case we say that $\alpha$ \emph{atomically universally defines} $S$. (Notice that \emph{universally definable} usually is defined analogously but instead of a single equation $\alpha$ one allows a Boolean formula of equations.)
It is clear that the center of a group is atomically universally definable by the expression $[X,Y]$. This generalizes as follows:

\begin{lemma}\label{lem:universallyd}
	 Let $G$ be a finite group.
	\begin{itemize}
		\item The Fitting group $\Fit(G)$ is atomically universally definable.
		\item If $N \leq H\leq G$ and $N$ is normal in $G$ and $H/N$ is atomically universally definable in $G/N$ and $N$ is atomically universally definable in $G$, then $H$ is atomically universally definable in $G$.
		\item All terms $\cU_iG$ of the upper Fitting series are atomically universally definable.
		\item If $H\leq G$ is inducible, then the centralizer $C_G(H) = \set{g\in G}{gh = hg \text{ for all } h \in H}$ is atomically universally definable.
	\end{itemize}
\end{lemma}
\begin{proof}
	By \cref{lem:Fitting}, the normal subgroup $\gen{g^G}$ generated by $g \in G$ is nilpotent if and only if $g \in \Fit(G)$. Therefore, $g \in \Fit(G)$ if and only if $\smcomm{M}{\gen{g^G}} = 1$ ($M$ as in \cref{sec:prelims} large enough), which, by \cref{lem:setcommutator}, is the case if and only if $\smcomm{M}{g^G}_{\Set} = 1$. Hence, the expression $[X^{Y_1}, \dots, X^{Y_{\!M}}]$ atomically universally defines $\Fit(G)$.

	Now, suppose that $\beta \in (G \cup \cX_\beta \cup \cX_\beta^{-1})^*$ with $\cX_\beta = \oneset{X,Y_1, \dots, Y_k}$ atomically universally defines $H/N$ in $G/N$ and that  $\alpha \in (G \cup \cX_\alpha \cup \cX_\alpha^{-1})^*$ with $\cX_\alpha = \oneset{Z,Y_{k+1}, \dots, Y_m}$ atomically universally defines $N$ in $G$. 
	Thus, $g \in H$ if and only if $\beta(g, Y_1, \dots, Y_k) \in N$ for all $Y_1, \dots, Y_k \in G$  and $h \in N$ if and only if $\alpha(h, Y_{k+1}, \dots, Y_m) =1$ for all $Y_{k+1}, \dots, Y_m \in G$. Hence, $\alpha(\beta(g, Y_1, \dots, Y_k), Y_{k+1}, \dots, Y_m) =1$ for all $Y_1, \dots, Y_m \in G$ if and only if $g \in H$ and so $H$ is atomically universally definable.

	The third point follows by induction from the first and second point.
	The fourth point is essentially due to \cite[Lemma 10]{HorvathS11}: 
	if $\beta$ is an expression inducing $H$, then $[X,\beta]$ atomically universally defines $C_G(H)$.
\end{proof}

\begin{lemma}\label{lem:univTAUT}
	Let $H \Nleq G$ be an atomically universally definable normal subgroup. Then \[\EQNID(G/H) \leq_{\mathrm{m}}^{\Ac0} \EQNID(G).\]	
\end{lemma}
\begin{proof}
	Denote $Q= G/H$.
	Let $\beta \in (G \cup \cX_\beta \cup \cX_\beta^{-1})^*$ with $\cX_\beta = \oneset{Z,Y_1, \dots, Y_k}$ atomically universally define $H$ and let $\alpha \in (Q \cup \cX \cup \cX^{-1})^*$ be an instance for $\EQNID(Q)$ (with $\cX \cap \cX_\beta = \emptyset$). Let $\tilde \alpha$ denote the expression obtained from $\alpha$ by replacing every constant of $Q$ by an arbitrary preimage in $G$. 
	Then $\sigma(\alpha) = 1$ in $Q$ for all assignments $\sigma: \cX \to Q$ if and only if  $\tilde \sigma(\tilde \alpha) \in H$ for all assignments $\tilde\sigma: \cX \to G$. By the choice of $\beta$, the latter is the case if and only if $\hat \sigma(\beta(\tilde \alpha,Y_1, \dots, Y_k)) =1$ for all assignments $\hat\sigma: \cX \cup \oneset{Y_1, \dots, Y_k} \to G$.
\end{proof}

\section{$G$-programs and AND-weakness}\label{sec:programs}

Let $G$ be a finite group.	
An $n$-input $G$-program of length $\ell$ with variables (input bits) from $\oneset{B_1, \dots, B_n}$ is a sequence
\[
P = \langle B_{i_1},a_1,b_1\rangle \langle B_{i_2},a_2,b_2\rangle \cdots \langle B_{i_\ell},a_\ell,b_\ell\rangle \in (\oneset{B_1, \dots, B_n} \times G \times G)^*.
\]
For a mapping $\sigma : \oneset{B_1, \dots, B_n} \to \{0,1\}$ (called an assignment) we define
$\sigma(P) \in G$ as the group element $c_1 c_2 \cdots c_\ell$, where $c_j = a_j$ if $B_{i_j} = 0$ and 
$c_j = b_j$ if $B_{i_j} = 1$ for all $1 \leq j \leq \ell$. 
We say that an $n$-input $G$-program $P$ \emph{computes} a function $f: \{0,1\}^n \to \{0,1\}$ if $P$ is over the variables $B_1, \dots, B_n$ and there is some $S\sse G$ such that $\sigma(P) \in S$ if and only if $f(\sigma) = 1$. 

\ProgSAT is the following problem: given a $G$-program $P$ with variables $B_1, \dots, B_n$, decide whether there is an assignment $\sigma:\oneset{B_1, \dots, B_n} \to G$ such that $\sigma(P)=1$.

\subparagraph*{The \AND-weakness conjecture.}
In \cite{BarringtonST90}, Barrington, Straubing and Thérien conjectured that, if $G$ is finite and solvable, every $G$-program computing the $n$-input \AND requires length exponential in $n$. This is called the \emph{\AND-weakness conjecture}.

Unfortunately, the term ``exponential'' seems to be a source of a possible misunderstanding: while often it means $2^{\Omega(n)}$, in other occasions it is used for $2^{n^{\Omega(1)}}$. 
Indeed, in \cite{GoldmannR02,BarringtonMMTT00}, the conjecture is restated as its \emph{strong version}: ``every $G$-program over a solvable group $G$ for the $n$-input \AND requires length $2^{\Omega(n)}$.'' However, already in the earlier paper \cite{BarringtonBR94}, it is  remarked that the $n$-input \AND can be computed by depth-$k$ \CC circuits of size $2^{\Oh(n^{1/(k-1)})}$ for every $k\geq 2$ (a \CC circuit is a circuit consisting only of $\MOD{m}$ gates for some $m \in \N$)~-- thus, disproving the strong version of the \AND-weakness conjecture.
For a recent discussion about the topic also referencing the cases where the conjecture actually is proved, we refer to \cite{Kompatscher19CC}.

In this section we provide a more detailed upper bound on the length of $G$-programs for the \AND function in terms of the Fitting length of $G$. We can view our upper bound as a refined version of the $2^{\Oh(n^{1/(k-1)})}$ upper bound for depth-$k$ \CC circuits. This is because, by \cite[Theorem 2.8]{McKenziePT91}, for every depth-$k$ \CC circuit family there is a fixed group $G$ of Fitting length $k$ (indeed, of derived length $k$) such that the $n$-input circuit can be transformed into a
 $G$-program of length polynomial in $n$. 

\iffull
The easiest variant to disprove the strong version of the \AND-weakness conjecture is a divide-and-conquer approach: Assume we can compute the $n$-input \AND by a \CC-circuit of size $2^n$ and depth $2$ (which is true by \cite{Barrington85}). Since we can decompose the $n$-input \AND as $\sqrt{n}$-input \AND of $\sqrt{n}$ many $\sqrt{n}$-input \AND{}s, we obtain a \CC circuit of depth $4$ and size roughly $2^{\sqrt{n}}$~-- or, more generally, a \CC circuit of depth $2k$ and size roughly $2^{\sqrt[k]{n}}$. The proof of \cref{prop:notANDweakrefined} uses a similar divide-and-conquer approach:
\fi

\begin{proposition}\label{prop:notANDweakrefined} Let $G$ be a finite solvable group and consider a strictly ascending series $1= H_0 \Nle H_1 \Nle \cdots \Nle H_m = G$ of normal subgroups where $H_i = \gamma_{k_i}(H_{i+1})$ with $k_i \in \N \cup\oneset{\infty}$ for $i \in \interval{1}{m-1}$ and $k_0 =\infty$. Denote $c = \abs{\set{i \in \interval{1}{m-1}}{k_i = \infty}}$ and $C = \prod_{k_i < \infty} (k_i + 1)$.
	
	Then the $n$-input \AND function can be computed by a $G$-program of length $\Oh(2^{D n^{1/c}})$ where $D = \frac{c}{C^{1/c}}$. More precisely, for every $n \in \N$ there is some $1\neq g \in G$ and a $G$-program $Q_n$ of length $\Oh(2^{D n^{1/c}})$ such that 
	\[\sigma(Q_n) = \begin{cases}
	g&\text{if } \sigma(B_1) = \cdots = \sigma(B_n) = 1,\\
	1& \text{otherwise.}
	\end{cases}\]
\end{proposition}
Clearly we have $c \leq d-1$ if $d$ is the Fitting length of $G$. The lower Fitting series is the special example of such a series where $H_i = \cL_{d-i}G$ and $k_i = \infty$ for all $i \in \oneset{0, \dots, d}$. Thus, we get the following corollary:

\begin{corollary}\label{cor:notANDweak}
	Let $G$ be a finite solvable group of Fitting length $d \geq 2$. Then the $n$-input \AND function can be computed by a $G$-program of length $2^{\Oh(n^{1/(d-1)})}$. 
\end{corollary}

\begin{example}
The symmetric group on four elements $S_4$ has Fitting length 3 with $S_4  \geq A_4 \geq C_2 \times C_2 \geq 1$ being both the upper and lower Fitting series. Therefore, we obtain a length-$\Oh(2^{2\sqrt{n})})$ program for the $n$-input \AND by \cref{prop:notANDweakrefined}. In particular, the strong version of the \AND-weakness conjecture does not hold for the group $S_4$. Note that according to \cite{BarringtonST90}, $S_4$ is the smallest group for which the $2^{\Omega(n)}$ lower bound from \cite{BarringtonST90} does not apply.

 On the other hand, consider the group $G = (C_3 \times C_3)\rtimes D_4$ where $D_4$ (the dihedral group of order eight) acts faithfully on $C_3 \times C_3$\footnote{This group can be found in the GAP small group library under the index $[72,40]$. It has been suggested as an example by Barrington (private communication).}. It has Fitting length two. Moreover, its derived subgroup $G' = (C_3 \times C_3)\rtimes C_2$ still has Fitting length two. Hence, we have a series $H_3 =G $, $H_2 = G' = \gamma_1 G$, $\,H_1 = \gamma_\infty G' = C_3 \times C_3$, and $H_0 = 1$. Therefore, we get an upper bound of $\Oh(2^{n/2})$ for the length of a program for the $n$-input \AND.
\end{example}

\begin{proof}[Proof of \cref{prop:notANDweakrefined}]
	We choose $K = (n/C)^{1/c}$. For simplicity, let us first assume that $K$ is an integer. Moreover, we  assume that $K$ is large enough such that $H_{i} =  \smcomm{K}{ H_{i+1}}$ holds whenever $k_i = \infty$ and that $K \geq k_i +1$ for all $k_i < \infty$.
	
	We define sets $A_i \sse G$ inductively by $A_{m} = G$ and $A_{i} = \smcomm{K}{A_{i+1}}_{\Set}$ if $k_i = \infty$ and $A_{i} = \smcomm{k_i+1}{A_{i+1}}_{\Set}$ if $k_i < \infty$.
	By \cref{lem:setcommutator} and induction it follows that $H_i = \gen{A_{i}}$ for all $i\in{0, \dots, m}$. Since $H_1 \neq 1$, we find a non-trivial element $g  \in A_{1}$. We can decompose $g $ recursively. For this, we need some more notation: for $\ell \in \interval{1}{m}$ consider the set of words \[V_\ell = \set{v = v_1\cdots v_{\ell-1} \in \interval{1}{K}^{\ell-1}}{v_i \leq k_{i} + 1 \text{ for all } i \in \interval{1}{\ell-1}}.\]
	We have $\abs{V_{m}} = C \cdot K^{c} = n$, so we can fix a bijection $\kappa\colon V_{m} \to \interval{1}{n}$. 
	
	Now, we can describe the recursive decomposition of $g = g_\epsilon$:
	\begin{itemize}
		\item $g_{v} = [g_{v1}, \dots, g_{vK}]$	for $v \in V_\ell$ with $k_\ell = \infty$, and
		\item $g_{v} = [g_{v1}, \dots, g_{v(k_\ell +1)}]$ for  $v \in V_\ell$ with $k_\ell < \infty$.
	\end{itemize}
This, in particular, we can view $g_\epsilon$ as a word over the $g_v$ for $v \in V_{m}$.
	
	For $v \in V_\ell$ we have $\abs{g_v} \leq \sum_{i=1}^{K} 2^{K+1-i}\abs{g_{vi}} \leq 2^{K+1} \max_{i }\abs{g_{vi}}$ whenever $k_\ell = \infty$ and $\abs{g_v} \leq 2^{k_\ell+2} \max_{i }\abs{g_{vi}}$ if $k_\ell < \infty$. Therefore, setting  $D = \frac{c}{C^{1/c}}$ we obtain by induction \[\abs{g_\eps} \leq 2^{\sum_{k_\ell < \infty} (k_\ell+2)} (2^{K+1})^{c} \in \Oh( 2^{Dn^{1/c}}).\]
	
	In order to obtain a $G$-program for the $n$-input \AND, we define $G$-programs $P_v$ for $v \in \bigcup_{\ell \leq m} V_\ell$. In the commutators we need also programs for inverses: for a $G$-program $P = \langle B_{i_1},a_1,b_1\rangle \langle B_{i_2},a_2,b_2\rangle \cdots \langle B_{i_\ell},a_\ell,b_\ell\rangle$  we set $P^{-1} = \langle B_{i_\ell},a_\ell^{-1},b_\ell^{-1}\rangle \cdots \langle B_{i_1},a_1^{-1},b_1^{-1}\rangle $. Clearly $(\sigma(P))^{-1} = \sigma(P^{-1})$ for all assignments $\sigma$.
	\begin{itemize}
		\item for  $v \in V_m$ we set  $P_v = \langle B_{\kappa(v)}, 1,g_v\rangle$,
		\item for  $v \in V_\ell $ with $1 \leq \ell < m$ we set $P_{v} = [P_{v1}, \dots, P_{vK}]$ if $k_{\ell} = \infty$, and
		\item for  $v \in V_\ell $ with $1 \leq \ell < m$ we set $P_{v} = [P_{v1}, \dots, P_{v(k_\ell +1)}]$ if $k_\ell < \infty$.
	\end{itemize}
	
	For $v \in V_\ell$ let  $V(v)$ denote the set of those words $w \in V_{m}$ having $v$ as a prefix.
	By induction we see that 
	\[\sigma(P_v) = \begin{cases}
	g_v&\text{if } \sigma(B_{\kappa(w)}) = 1 \text{ for all } w\in V(v),\\
	1& \text{otherwise.}
	\end{cases}\]
	This shows the correctness of our construction.
	
	It remains to consider the case that $(n/C)^{1/c}$ is not an integer. Then we set $K = \ceil{(n/C)^{1/c}}$. It follows that $\abs{V_{m}} = C \cdot K^{c} \geq n$, so we can fix a bijection $\kappa \colon U\to \interval{1}{n}$ for some subset $U \sse  V_{m}$. We still have $\abs{g_\eps} \leq 2^{\sum_{k_i < \infty} (k_i+1)} (2^{K+1})^{c} \in \Oh (2^{cK}) = \Oh( 2^{Dn^{1/c}})$ with $D$ as above.	
	This concludes the proof of \cref{prop:notANDweakrefined}.
\end{proof}

\begin{remark}
In the light of \cref{prop:notANDweakrefined} it is natural to ask for a refined version of the \AND-weakness conjecture. A natural candidate would be to conjecture that every $G$-program for the $n$-input \AND has length $2^{\Omega(n^{1/(d-1)})}$ where $d$ is the Fitting length of $G$.

However, this also weaker version of the \AND-weakness conjecture is wrong! 
 Indeed, in \cite[Section 2.4]{BarringtonBR94} Barrington, Beigel and Rudich show that the $n$-input \AND can be computed by circuits using only $\MOD{m}$ gates of depth 3 and size $2^{\Oh(n^{1/r} \log n)}$ where $r$ is the number of different prime factors of $m$. Translating the circuit into a $G$-program yields a group $G$ of Fitting length 3. Since there is no bound on $r$, we see that there is no lower bound on the exponent $\delta$ such that there are $G$-programs of length $2^{\Oh(n^\delta)}$ for the $n$-input \AND in groups of Fitting length 3. 
\iffull
While this does not yield smaller \CC circuits or shorter $G$-programs than the approach of \cref{prop:notANDweakrefined} allows, it shows that the divide-and-conquer technique on which \cref{prop:notANDweakrefined} relies is not always the best way for constructing small programs for \AND. 
\fi
\end{remark} 

In \cite{HansenK10} it is shown that the \AND function can be computed by probabilistic \CC circuits using only a logarithmic number of random bits, which ``may be viewed as evidence contrary to the conjecture'' \cite{HansenK10}.
 In the light of this, we do not feel confident to judge which form of the \AND-weakness conjecture might be true. The following version seems possible.

\begin{conjecture}[\AND-weakness \cite{BarringtonST90}]\label{conj:andweak}
	Let $G$ be finite solvable. Then every $G$-program for the $n$-input \AND has length $2^{n^{\Omega(1)}}$.
\end{conjecture}
Notice that \cite[Theorem 2]{BarringtonMMTT00} (if $G$ is \AND-weak, \ProgSAT over $G$ can be decided in quasi-polynomial time) still holds with this version of the \AND-weakness conjecture.

\section{Reducing \KColoring{C} to equations}\label{sec:reduction}

In this section we describe the reduction of \KColoring{C} to $\EQNSAT(G)$ and $\EQNID(G)$ in the spirit of \cite{GoldmannR02,Kompatscher19}. For this, we rely on the fact that $G$ has some normal subgroups meeting some special requirements. In \cref{sec:consequences}, we show that all sufficiently complicated finite solvable groups meet the requirements of \cref{thm:main}.

For a normal subgroup $H \Nleq G$ and $g\in G$, we define
$\eta_g(H) = \fitcomm{H}{g^G}$. Recall that $M $ is chosen large enough such that $\mcomm{X}{M}{Y} = \mcomm{X}{i}{Y}$ for all $i \geq M$ and all $X,Y\sse G$ with $X^G=X$ and $Y^G=Y$.
Since $H$ is normal, we have $\eta_g(H)\leq H$ and $\eta_g(H)$ is normal in $G$.

\begin{lemma}\label{lem:repeateta}\label{lem:Hsubgroup}
	Let $H\Nleq G$ be a normal subgroup and $g,h \in G$. Then 
	\begin{enumerate}
		\item $\eta_g(\eta_g(H)) = \eta_g(H)$, and
		\item $\eta_{gh}(H) \leq  \eta_g(H)\eta_h(H)$, and
		\item $\FitL(\eta_{gh}(H)) \leq  \max\oneset{\FitL(\eta_g(H)), \FitL(\eta_h(H))}$.
	\end{enumerate}
\end{lemma}
\begin{proof}
	We use the fact that $M$ is chosen such that  $\mcomm{X}{M}{Y} = \mcomm{X}{i}{Y}$ for all  $i \geq M$ and all $X,Y \sse G$ with $X^G=G$ and $Y^G=Y$:	
\begin{align*}
	\eta_g(H) &= \mcomm{H}{M}{g^G}= \mcomm{H}{2M}{g^G}=  \mcomm{\mcomm{H}{M}{g^G}}{M}{g^G\vphantom{k^k}} = \eta_g(\eta_g(H)).
\end{align*}
The second point follows with the same kind of argument:
\begin{align*}
\eta_{gh}(H) &= [H,\,_{2M} (gh)^G]\leq [H,\,_{2M} \gen{g^G \cup h^G}]\\
& = \gen{[H,\,_{2M} g^G \cup h^G]_{\Set}} \tag{by \cref{lem:setcommutator}}\\
&\leq \eta_g(H)\eta_h(H).
\end{align*}
The last step is because each of the commutators in $[H,\,_{2M} g^G \cup h^G]_{\Set}$ either contains at least $M$ terms from $g^G$ and, thus, is in $\eta_g(H)$ or it contains at least $M$ terms from $h^G$. 

The third point is an immediate consequence of the second point and \cref{lem:Fitting}.
\end{proof}

\begin{lemma}\label{lem:Kinducible}
	Suppose that $K \Nleq G$ is a normal subgroup satisfying $\eta_g(K) = K$ for some $g \in G$. Then $K$ is inducible.
\end{lemma}
\begin{proof}
Because $\eta_{g}(K) = K$ for some $g\in G$ implies that $K = [K,G]$, it follows from \cref{lem:inducible} that $K$ is inducible.
\end{proof}

\begin{theorem}\label{thm:main}
	Let $G$ be a finite solvable group of Fitting length three and assume there are normal subgroups $K\Nleq H\Nleq  G$ such that $\FitL(K) = 2$, $\cU_{2}G \leq H$, and $\abs{G/H} \geq 3$. Moreover, assume that
	\begin{enumerate}[(I)]
		\item for all $g \in G  \setminus H$ we have $\eta_g(K) = K$,\label{assumption1}
		\item for all $h \in H$ we have $\FitL(\eta_h(K)) \leq 1$.\label{assumption2}
	\end{enumerate}
Then $\EQNSAT(G)$ and $\EQNID(G)$ cannot be decided in deterministic time $2^{o(\log^2N)}$ under ETH where $N$ is the length of the input expression. In particular, $\EQNSAT(G)$ and $\EQNID(G)$ are not in \P under ETH.
\end{theorem}

 \newcommand{\numBatches}{R}
\newcommand{\indexBatchOut}{r}
\newcommand{\indexBatchIn}{s}
\newcommand{\indexK}{k}
\newcommand{\indexMgamma}{\nu}
\newcommand{\indexMdelta}{\mu}
\newcommand{\numKInd}{T}
\newcommand{\indexKInd}{t}

\subparagraph*{Proof outline.}	The crucial observation for this theorem is the same as for \cref{prop:notANDweakrefined}: that, roughly speaking, the $n$-input \AND can be decomposed into the conjunction of $\sqrt{n}$ many $\sqrt{n}$-input \AND{}s. We use this observation in order to reduce the \KColoring{C} problem to \EQNSAT.
More precisely, given a graph $\Gamma$ with $n$ vertices and $m$ edges, we construct an expression $\delta$ and an element $\tilde h \in G$ such that 
\begin{enumerate}[(A)]
	\item the length of $\delta$ is in $2^{\Oh(\sqrt{m+n})}$, \label{pointA}
	\item $\delta$ can be computed in time polynomial in its length,\label{pointB}
	\item $\delta =\tilde h$ is satisfiable if and only if $\Gamma$ has a valid $C$-coloring, and\label{pointC}
	\item $\sigma(\delta) = 1$ holds for all assignments $\sigma$ if and only if $\Gamma$ does \emph{not} have a valid $C$-coloring. \label{pointD}
\end{enumerate}
 For the number of colors we use $C = \abs{G/H}$. 
 Let $N$ denote the input length for \EQNSAT (resp.\ \EQNID).
 A $2^{o(\log^2N)}$-time algorithm for \EQNSAT (resp.\ \EQNID), thus, would imply a $2^{o(n+m)}$-time algorithm for \KColoring{C} contradicting ETH. 
Hence, it is enough to show points (\ref{pointA})--(\ref{pointD}).

 In order to construct the expression $\delta$, we assign a variable $X_i$ to every vertex $v_i$ of $\Gamma$. Every assignment $\sigma$ to the variables $X_i$ will give us a coloring $\chi_\sigma$ of $\Gamma$ (to be defined later). During the proof, we also introduce some auxiliary variables. The aim is to construct $\delta$ in a way that an assignment $\sigma$ to the variables $X_i$
 can be extended to a satisfying assignment for $\delta =\tilde h$ if and only if $\chi_\sigma$ is a valid coloring of $\Gamma$ (see \cref{lem:reductioncorrect}).

We start by grouping the edges into roughly $\sqrt{m}$ batches of $\sqrt{m}$ edges each. For each batch of edges, we construct an expression $\gamma_\indexBatchOut$ (where $\indexBatchOut$ is the number of the batch) such that for every assignment $\sigma$ to the variables $X_i$ we have
 \begin{itemize}
 	 \item if  $\chi_\sigma$ assigns the same color to two endpoints of an edge in the $\indexBatchOut$-th batch, then for every assignment to the auxiliary variables, $\gamma_\indexBatchOut$ evaluates to something in $\cU_1K$,
 	\item otherwise, for every element $h \in K$, there is an assignment to the auxiliary variables such that $\gamma_\indexBatchOut$ evaluates to $h$.
 \end{itemize}
A more formal statement of this can be found in \cref{lem:assignmentextension1}.
The expression $\delta$ combines all the $\gamma_\indexBatchOut$ as an iterated commutator such that if one of the $\gamma_\indexBatchOut$ evaluates to something in $\cU_1K$, then $\delta$ evaluates to $1$, and, otherwise, there is some assignment to the auxiliary variables such that $\delta$ evaluates to the fixed element $\tilde h$.

\begin{proof}
Let $C = \abs{G/H}$.
Let us describe how the \KColoring{C} problem for a given graph $\Gamma=(V,E)$ is reduced to an instance of \EQNSAT (resp.\ \EQNID). 
 We denote $V= \oneset{v_1, \dots, v_n}$. For every vertex $v_i$ we introduce a variable $X_i$ and we set $\cX = \oneset{X_1, \dots, X_n}$. 
 By fixing a bijection $\abs{G/H} \to \interval{1}{C}$, we obtain a correspondence between assignments $\cX \to G$ and colorings $V \to \interval{1}{C}$ (be aware that it is \emph{not} one-to-one). During the construction we will also introduce a set $\cY$ of auxiliary variables. As outlined above, the idea is that an assignment $\cX \to G$ represents a valid coloring if and only if there is an assignment to the auxiliary variables under which the equation evaluates to a non-identity element.

For each edge $\oneset{v_i,v_j} \in E$, we introduce one edge gadget $X_iX_j^{-1}$ (it does not matter which one is the positive variable).
Now, we group these gadgets into $\numBatches$ batches of $\numBatches$ elements each (if the number of gadgets is not a square, we duplicate some gadgets)~-- \ie we choose $\numBatches = \ceil{\sqrt{m}\;\!}$. How the gadgets exactly are grouped together does not matter.

For $\indexBatchOut \in \interval{1}{\numBatches}$ and $\indexK\in \interval{1}{\abs{K}}$ let $\alpha_{\indexBatchOut,\indexK}$ be an expression which induces $K$ (\ie all $\alpha_{\indexBatchOut,\indexK}$ are the same expressions but with disjoint sets of variables). Such expressions exist by \cref{lem:Kinducible}.
Let the variables of $\alpha_{\indexBatchOut,\indexK}$ be $Y_{\indexBatchOut,\indexK,\indexKInd}$ for $\indexKInd\in \interval{1}{\numKInd} $ for some $\numKInd \in \N$. 
 Moreover, we introduce more auxiliary variables $Z_{\indexBatchOut,\indexK,\indexBatchIn,\indexMgamma}$ for $\indexBatchOut \in \interval{1}{\numBatches}$, $\indexK\in \interval{1}{\abs{K}}$, $\indexBatchIn \in \interval{1}{\numBatches}$, and $ \indexMgamma \in \interval{1}{M}$  (recall that $M$ is chosen such that, in particular, $\fitcomm{H_1}{H_2} = \mcomm{H_1}{M+1}{H_2}$ for arbitrary normal subgroups $H_1, H_2$ of $G$) and we set 
\[\cY'_\indexBatchOut = \set{\vphantom{\big(} Z_{\indexBatchOut,\indexK,\indexBatchIn,\indexMgamma},\; Y_{\indexBatchOut,\indexK,\indexKInd}}{\indexK\in \interval{1}{\abs{K}}, \indexBatchIn \in \interval{1}{\numBatches}, \indexMgamma \in \interval{1}{M}, \indexKInd \in \interval{1}{\numKInd}  }.\] 
Let $\beta_{\indexBatchOut,1}, \dots, \beta_{\indexBatchOut,\numBatches}$ be the gadgets of the $\indexBatchOut$-th batch for some $\indexBatchOut \in \interval{1}{\numBatches}$. We define
\begin{align}
\gamma_\indexBatchOut = 
\prod_{\indexK=1}^{\abs{K}} \left[ \alpha_{\indexBatchOut,\indexK},  \beta_{\indexBatchOut,1}^{Z_{\indexBatchOut,\indexK,1,1}}, \dots, \beta_{\indexBatchOut,1}^{Z_{\indexBatchOut,\indexK,1,M}}, \dots,\beta_{\indexBatchOut,\numBatches}^{Z_{\indexBatchOut,\indexK,\numBatches,1}}, \dots, \beta_{\indexBatchOut,\numBatches}^{Z_{\indexBatchOut,\indexK,\numBatches,M}}\right].\label{eq:gammak}
\end{align}

 We do this for every batch of gadgets. The following observation is crucial:

 \begin{lemma}\label{lem:assignmentextension1}
 	Let $\sigma\colon \cX \to G$ be an assignment and let $\indexBatchOut \in \interval{1}{\numBatches}$. 
\begin{itemize}			
 	\item	If $\sigma(\beta_{\indexBatchOut,\indexBatchIn}) \in G\setminus H$ for all $\indexBatchIn$, then
 		$\displaystyle \set{\vphantom{k^h}(\sigma \cup \sigma')(\gamma_\indexBatchOut)}{\sigma' : \cY_\indexBatchOut' \to G} = K,$
 	\item	Otherwise,
 		$\displaystyle\set{\vphantom{k^h}(\sigma \cup \sigma')(\gamma_\indexBatchOut)}{\sigma' : \cY_\indexBatchOut' \to G} \leq \cU_1K.$
 	\end{itemize}
 \end{lemma}
 
 \newcommand{\comMu}{, \kern.1em_{M}\,\kern.1em}
 
 \begin{proof}
 	By construction, we have $(\sigma \cup \sigma')(\alpha_{\indexBatchOut,\indexK}) \in K$ for all $\indexBatchOut$ and $\indexK$ and all assignments $\sigma $ and $\sigma'$. Since $K$ is normal, it follows that $(\sigma \cup \sigma')(\gamma_\indexBatchOut) \in K$ for all assignments $\sigma $ and $\sigma'$.
 	
 	Consider the case that $g_\indexBatchIn\coloneqq \sigma(\beta_{\indexBatchOut,\indexBatchIn}) \in G\setminus H$ for all $\indexBatchIn \in \interval{1}{\numBatches}$. By assumption (\ref{assumption1}), we have $K = \eta_{g_1}(K) = \eta_{g_2}(\eta_{g_1}(K)) = \cdots = \eta_{g_\numBatches} \dots \eta_{g_2}(\eta_{g_1}(K))\cdots)  $. By \cref{lem:setcommutator}, it follows that 
 	$K = \gen{ [K\comMu g_1^G, \dots\comMu g_\numBatches^G]_{\Set}}.$
 	Since $1 \in  [K\comMu  g_1^G, \dots\comMu g_\numBatches^G]_{\Set}$ and every element in $K$ can be written as a product of length at most $\abs{K}$ over any generating set, we conclude  $K = \left([K\comMu g_1^G,\dots\comMu g_\numBatches^G]_{\Set}\right)^{\abs{K}}$. This is exactly the form how $\gamma_\indexBatchOut$ was defined in \cref{eq:gammak} (recall that $\alpha_{\indexBatchOut,\indexBatchIn}$ can evaluate to every element of $K$). Therefore, for each $h \in K$, there is an assignment $\sigma'\colon \cY_\indexBatchOut' \to G$ such that $(\sigma \cup \sigma')(\gamma_\indexBatchOut) = h$.

 	On the other hand, let $g_\indexBatchIn\coloneqq \sigma(\beta_{\indexBatchOut,\indexBatchIn}) \in H$ for some $\indexBatchIn$.
 	Then, by assumption (\ref{assumption2}) we have $\FitL(\eta_{g_\indexBatchIn}(K)) \leq 1$.
 	 Since $(\sigma \cup \sigma')(\gamma_\indexBatchOut) \in \eta_{g_\indexBatchIn}(K)$, we obtain $(\sigma \cup \sigma')(\gamma_\indexBatchOut) \in \cU_1K$ by \cref{lem:Fitting}.
 \end{proof}

Now, for every set of auxiliary variables $\cY_\indexBatchOut'$ we introduce $M$ disjoint copies, which we call $\cY_\indexBatchOut^{(\indexMdelta)}$ for $\indexMdelta \in \interval{1}{M}$.
 We write $\gamma_\indexBatchOut^{(\indexMdelta)}$ for the copy of $\gamma_\indexBatchOut$ where the variables of $\cY_\indexBatchOut'$ are substituted by the corresponding ones in  $\cY_\indexBatchOut^{(\indexMdelta)}$ (the variables $\cX$ are shared over all $\gamma_\indexBatchOut^{(\indexMdelta)}$). We set
\[\delta  = \bigl[\gamma_1^{(1)}, \dots , \gamma_1^{(M)}, \dots ,   \gamma_\numBatches^{(1)}, \dots , \gamma_\numBatches^{(M)} \bigr].\]
Finally, fix some $\tilde h \in K \setminus 1$ with $\tilde h \in \smcomm{M\kern-.05em\cdot\kern-.05em\numBatches}{K}_{\Set}$ and set $\cY =  \bigcup_{\indexBatchOut,\indexMdelta} \cY_\indexBatchOut^{(\indexMdelta)}$.

\begin{lemma}\label{lem:assignmentextension2}
Let $\sigma\colon \cX \to G$ be an assignment.
If $\sigma(\beta_{\indexBatchOut,\indexBatchIn}) \in G\setminus H$ for all $\indexBatchOut$ and $\indexBatchIn$, then there is some assignment $\sigma'\colon \cY \to G$ such that $(\sigma \cup \sigma')(\delta) = \tilde h$. Otherwise  $(\sigma \cup \sigma')(\delta) = 1$ for all $\sigma'\colon \cY \to G$.
\end{lemma}

\begin{proof}
	If $\sigma(\beta_{\indexBatchOut,\indexBatchIn}) \in G\setminus H$ for all $\indexBatchOut$ and $\indexBatchIn$, then by \cref{lem:assignmentextension1}, 
	$\big\{(\sigma \cup \sigma')(\gamma_\indexBatchOut^{(\indexMdelta)}) \:\big|\: \sigma' : \cY_\indexBatchOut^{(\indexMdelta)} \to G \big\} = K$ 
	for all $\indexBatchOut\in \interval{1}{\numBatches}$ and $\indexMdelta\in \interval{1}{M}$. Hence, since we chose the auxiliary variables $\cY_\indexBatchOut^{(\indexMdelta)}$ to be all disjoint, we obtain
	\[ \tilde h \in \smcomm{M\kern-.05em\cdot\kern-.05em\numBatches}{K}_{\Set} \sse \set{(\sigma \cup \sigma')(\delta)}{\sigma' : \cY_\indexBatchOut^{(\indexMdelta)} \to G}.\]
	
On the other hand, if $\sigma(\beta_{\indexBatchOut,\indexBatchIn}) \in H$, then, by \cref{lem:assignmentextension1}, 
for all $\sigma'\colon \cY \to G$ and all $\indexMdelta\in \interval{1}{M}$ we have  $(\sigma \cup \sigma')(\gamma_\indexBatchOut^{(\indexMdelta)}) \in \cU_1K$. Hence, $(\sigma \cup \sigma')(\delta ) \in \smcomm{M}{\cU_1K}  = 1$.
\end{proof}

Now we are ready to define our equation as $\delta\tilde h^{-1} $ for the reduction of \KColoring{C} to $ \EQNSAT(G)$ and  $\delta $ for the reduction to $ \EQNID(G)$. 

The final step is to show points  (\ref{pointA})--(\ref{pointD}) from above.

For (\ref{pointA}) observe that the length of $\gamma_\indexBatchOut$ is $\Oh(2^{M\cdot\numBatches})$ for all $\indexBatchOut$. Thus, the length of $\delta$ is  $\Oh(2^{M\cdot\numBatches})\cdot \Oh(2^{M\cdot\numBatches}) \sse 2^{\Oh(\numBatches)} = 2^{\Oh(\sqrt{m})}$ as desired. Point (\ref{pointB}) is straightforward from the construction of $\delta$.

In order to see (\ref{pointC}) and (\ref{pointD}), we use \cref{lem:assignmentextension2} to prove another lemma.
We fix a bijection $\xi: G/H\to \interval{1}{C}$.
For an assignment $\sigma: \cX \to G$, we define a corresponding coloring $\chi_\sigma : V \to \interval{1}{C}$ by $\chi_\sigma(v_i) = \xi(\sigma(X_i)H)$.

\begin{lemma}\label{lem:reductioncorrect}
	Let  $\sigma: \cX \to G$ be an assignment. Then
	\begin{itemize}
		\item if $\chi_\sigma$ is valid, then there is an assignment $\sigma' :\cY \to G$ such that $(\sigma \cup \sigma')(\delta) = \tilde h\neq 1$,
		\item if $\chi_\sigma$ is \emph{not} valid, then for all assignments $\sigma' :\cY \to G$ we have $(\sigma \cup \sigma')(\delta) = 1$.
	\end{itemize}
\end{lemma}

\begin{proof}
Let $\chi_\sigma$ be a valid coloring. First, observe that the gadgets all evaluate to some element outside of $H$ under $\sigma$. This is because, if there is a gadget $X_iX_j^{-1}$ that means that $\oneset{v_i, v_j} \in E$ and so $\chi_\sigma(v_i) \neq \chi_\sigma(v_j)$; hence, $\sigma(X_i) \neq \sigma(X_j)$ in $G/H$ (since $\xi$ is a bijection).
Therefore, by \cref{lem:assignmentextension2}, it follows that $\delta$ evaluates to $\tilde h$ under some proper assignment for $\cY$.

On the other hand, if $\chi_\sigma$ is not a valid coloring, then there is an edge $\oneset{v_i, v_j} \in E$ with $\chi_\sigma(v_i) = \chi_\sigma(v_j)$. Then we have $\sigma(X_i)H = \sigma(X_j)H$. Hence, by \cref{lem:assignmentextension2}, we obtain that $(\sigma_\chi\cup \sigma')(\delta)=1$ in $G$ for every $\sigma': \cY \to G$. 
\end{proof}
This concludes the proof of \cref{thm:main}.
\end{proof}

\section{Consequences}\label{sec:consequences}

In this section we derive our main result Corollary~\ref{cor:mainIntro}. We start again with a lemma.

\begin{lemma}\label{lem:KHG}
	For every finite solvable, non-nilpotent group $G$ of Fitting length $d$, there are proper normal subgroups $K\Nleq H \Nle G$ with $\FitL(K) = d -1$ and $\cU_{d -1}G \leq H$ such that
	\begin{itemize}
		\item for all $g \in G  \setminus H$ we have $\eta_g(K) = K$,
		\item for all $h \in H$ we have $\FitL(\eta_h(K)) < \FitL(K)$.
	\end{itemize} 
\end{lemma}

The construction for \cref{lem:KHG} resembles the ones in Lemmas 5 and 6 of \cite{Kompatscher19}. However, while in \cite{Kompatscher19} a minimal normal subgroup $N$ of a quotient $G/K$ is constructed such that $r_g$ with $r_g(x) = [x,g]$ is an automorphism of $N$ (and $N$ is abelian), in our case this is not enough since we need to apply commutator constructions to our analog of $N$ in the spirit of the divide-and-conquer approach of \cref{prop:notANDweakrefined}.

\begin{proof}
	Let $g_1 \in G \setminus\cU_{d-1} G$ where $d$ is the Fitting length of $G$. We construct a sequence of normal subgroups $K_1, K_2, \ldots $ of $ G$ as follows: we set 
	$K_1 = \eta_{g_1}(G) $. By \cref{lem:addGincomm}, $ K_1= \gamma_\infty\gen{g_1^G}$, so it has Fitting length $d-1$.
	
	Now, while there is some $g_i \in G$ such that $\eta_{g_i}(K_{i-1}) < K_{i-1}$ and $\FitL(\eta_{g_i}(K_{i-1})) = \FitL(K_{i-1})$, we set $K_i = \eta_{g_i}(K_{i-1})$ and continue.
	Since $K_i$ is a proper subgroup of $K_{i-1}$, this process eventually terminates. We call the last term $K$. 
	We claim that $K$ satisfies the statement of \cref{lem:KHG}.	
	By construction for every $g \in G$ one of the two cases 
	\begin{itemize}
		\item $\eta_g(K) = K$ or
		\item $\FitL(\eta_{g}(K)) < \FitL(K)$
	\end{itemize}	
	applies. Moreover, since $K = \eta_g(K')$ for some $K'\leq G$ and some $g \in G$, we have $K = \eta_g(K') = \eta_g(\eta_g(K')) = \eta_g(K)$ by \cref{lem:repeateta} (i). 
	By \cref{lem:Hsubgroup} (iii), the elements $\set{h \in G}{\FitL(\eta_h(K)) < \FitL(K)}$ form a subgroup $H$ of $G$. Clearly $H$ is normal (by the definition of $\eta_h$) and $K \leq \cU_{d -1}G \leq H$ because $\FitL(\mcomm{K}{M}{\cU_{d -1}G}) = \FitL(K) - 1$. Since there is some $g \in G$ with $K = \eta_g(K)$, we have $H \neq G$.
\end{proof}	
Be aware that $K$ depends on the order the $g_i$ were chosen. Indeed, if $G$ is a direct product of two groups $G_1$ and $G_2$ of equal Fitting length, then $K$ will either be contained in $G_1$ or in $G_2$~-- in which factor depends on the choice of the $g_i$.

\begin{theorem}[Corollary~\ref{cor:mainIntro}]\label{thm:main2}
	\iffull
	Let $G$ be a finite solvable group meeting one of the following conditions:
	\begin{enumerate}
		\item $\FitL(G) = 3$ and $\abs{G/\cU_2G}$ has a prime divisor 3 or greater (\ie $G/\cU_2G$ is not a 2-group),
		\item $\FitL(G) \geq 4$.		
	\end{enumerate}
\else
	Let $G$ be a finite solvable group such that either $\FitL(G) = 3$ and $\abs{G/\cU_2G}$ has a prime divisor 3 or greater (\ie $G/\cU_2G$ is not a 2-group)
	or $\FitL(G) \geq 4$.		
\fi
	Then $\EQNSAT(G)$ and $\EQNID(G)$ cannot be decided in deterministic time $2^{o(\log^2N)}$ under ETH. In particular, $\EQNSAT(G)$ and $\EQNID(G)$ are not in \P under ETH.
\end{theorem}

\begin{proof}
	Consider the case that $G$ has Fitting length 3 and $\abs{G/\cU_2G}$ has a prime divisor 3 or greater.
	Let $2^\nu$ for some $\nu \in \N$ be the greatest power of two dividing $\abs{G/\cU_2G}$. Then, the subgroup $\wt G$ generated by all $2^\nu$-th powers is normal and it is not contained in $\cU_2G$. Therefore, by \cref{lem:Fitting} it has Fitting length $3$ as well. 
	Also, by \cref{lem:Fitting}, we know that $\cU_2\wt G =\wt G \cap  \cU_2G $. Hence, $\wt G/\cU_2\wt G$ is a subgroup of $G/\cU_2G$. Moreover, since $\wt G$ is generated by $2^\nu$-th powers, the generators of $\wt G$ have odd order in $\wt G/\cU_2\wt G$. Since $\wt G/\cU_2\wt G$ is nilpotent, it follows that $|\wt G/\cU_2\wt G|$ is odd (recall that a nilpotent group is a direct product of $p$-groups). 
	
	Since $\wt G$ is inducible in $G$, by \cref{lem:inducibleEQN}, it suffices to show that $\wt G$ satisfies the requirements of \cref{thm:main}. For this, we use \cref{lem:KHG}, which gives us normal subgroups $K \Nleq H \Nle \wt G$ with $\cU_2\wt G \leq H$,  $\FitL(K) = 2$ and such that for all $g \in\wt G  \setminus H$ we have $\eta_g(K) = K$,
	and for all $h \in H$ we have $\FitL(\eta_h(K)) \leq 1$.
	
	It only remains to show that $|\wt G/H| \geq 3$. Since $H\neq \wt G$ and $|\wt G/H|$ is odd, this holds trivially. Thus, both $\EQNSAT(G)$ and $\EQNID(G)$ are not in \P  under ETH if $G$ has Fitting length 3 and $\abs{G/\cU_2G}$  a prime divisor 3 or greater.

	\medskip
	The second case can be reduced to the first case as follows: Assume that $G$ has Fitting length $d \geq 4$. If $\abs{G/\cU_{d-1} G}$ has a prime factor $3$ or greater, we can apply the Fitting length 3 case to $G /\cL_3 G$ for \EQNSAT and to $G /\cU_{d-3} G$ for \EQNID. By \cref{lem:inducible} and \cref{lem:inducibleEQN} this implies the corollary for \EQNSAT. For \EQNID, the statement follows form \cref{lem:universallyd} and \cref{lem:univTAUT}.
	
	On the other hand, if $\abs{G/\cU_{d-1} G} = 2^\nu$ for some $\nu \geq 1$, as in the first case, we consider the subgroup $\wt G$ generated by all $2^\nu$-th powers. Then the index of $\wt G$ in $G$ is again a power of two (since the order of every element in $G/\wt G$ is a power of two). Moreover, $\wt G \leq \cU_{d-1} G$ and, by \cref{lem:Fitting}, we have
	\[
	\wt G/\cU_{d-2} \wt G = \wt G/(\cU_{d-2}G \cap \wt G )  \cong (\wt G \cdot \cU_{d-2} G)/\cU_{d-2} G \leq \cU_{d-1} G/\cU_{d-2} G.
	\]

	Now, $\abs{\cU_{d-1} G/\cU_{d-2} G}$ cannot be a power of two because, otherwise,  $G/\cU_{d-2} G$ would be a 2-group and, thus, nilpotent~-- contradicting the fact that the upper Fitting series is a shortest Fitting series. Since the index of $ \wt G $ in $\cU_{d-1} G$ is a power of two, we see that $\wt G \not\sse \cU_{d-2} G$ and that the index of $\cU_{d-2} \wt G$ in  $\wt G$ has a prime factor other than 2.
	Therefore, we can apply the Fitting length 3 case to  $\wt G /\cL_3 \wt G$ (resp.\ $\wt G /\cU_{d-3} \wt G$).
\end{proof}

\subparagraph*{The case that $G/\cU_2G$ is a 2-group.}
As mentioned above, in the recent paper \cite{IdziakKK20} Idziak, Kawa\l ek, and Krzaczkowski proved a $2^{\Oh(\log^2(n))}$-lower bound under ETH for $\EQNSAT(S_4)$. They apply a reduction of \SAT to $\EQNSAT(S_4)$. 
Instead of using commutators to simulate conjunctions in the group, 
the more complicated logical function $(X,Y_1,Y_2,Y_3) \mapsto X\land(Y_1\lor  Y_2\lor Y_3)$ is encoded into the group. Indeed, under suitable assumptions on the group and the range of the variables, both the expressions $w(X,Y_1,Y_2,Y_3) = X^8[X,Y_1,Y_2,Y_3]$ (see \cite{Kompatscher19}) and $s(X,Y_1,Y_2,Y_3) = X\;\![X,Y_1,Y_2,Y_3]^{-1}$ (see \cite{GorazdK10}~-- referred to by \cite{IdziakKK20}) simulate this logical function. A new paper unifying our approaches and proving \cref{thm:main2} for \emph{all} groups of Fitting length 3 can be found in \cite{IdziakKKW20arxiv}.

\subparagraph*{Consequences for \ProgSAT.}
We have $ \EQNSAT(G)\leq_{\mathrm{m}}^{\Ac0} \ProgSAT(G)$  for every finite group $G$ by \cite[Lem.~1]{BarringtonMMTT00} (while not explicitly stated, it is clear that this reduction is an \Ac0-reduction). 
Thus, by \cref{thm:main}, $\ProgSAT(G)$ is not in \P under ETH if $G$ is of Fitting length at least 4 or $G$ is of Fitting length 3 and $G/\cU_{2}G$ is not a $2$-group.

\subparagraph*{Small groups for which \cref{thm:main2} gives a lower bound.}
In \cite{Horvath15} lists of groups are given where the complexity of $\EQNSAT$ and $\EQNID$ is unknown. The paper refers to a more comprehensive list available on the author's website \url{http://math.unideb.hu/horvath-gabor/research.html}. We downloaded the lists of groups and ran tests in GAP for which of these groups \cref{thm:main2} provides lower bounds. In the list with unknown complexity for $\EQNID$ there are 2331 groups of order less than 768 out of which 1559 are of Fitting length three or greater. \cref{thm:main2} applies to 22 of them: 3 groups of Fitting length 4 and 19 groups $G$ of Fitting length 2 where $G/\cU_{2}G$ is not a 2-group. A list of the groups for which we could prove lower bounds can be found in \cref{tab:GAPresults}.

\begin{table}
	\caption{Groups up to order 767 for which \cref{thm:main2} gives lower bounds.}\label{tab:GAPresults}
	
	{\small	\begin{tabular}{|c|c|l|}
			\parbox{2.2cm}{Index in Small Groups Library} & \parbox{.92cm}{Fitting length} & GAP Structure description\\	
			\hline
			[ 168, 43 ] & 3 & (C2 x C2 x C2) : (C7 : C3)\\{} 
			[ 216, 153 ] & 3 & ((C3 x C3) : Q8) : C3\\{} 
			[ 324, 160 ] & 3 & ((C3 x C3 x C3) : (C2 x C2)) : C3\\{}  
			[ 336, 210 ] & 3 & C2 x ((C2 x C2 x C2) : (C7 : C3))\\{}  
			[ 432, 734 ] & 4 & (((C3 x C3) : Q8) : C3) : C2\\{}  
			[ 432, 735 ] & 3 & C2 x (((C3 x C3) : Q8) : C3)\\{}  
			[ 504, 52 ] & 3 & (C2 x C2 x C2) : (C7 : C9)\\{}  
			[ 504, 158 ] & 3 & C3 x ((C2 x C2 x C2) : (C7 : C3))\\{}  
			[ 600, 150 ] & 3 & (C5 x C5) : SL(2,3)\\{} 
			[ 648, 531 ] & 3 & C3 . (((C3 x C3) : Q8) : C3) = (((C3 x C3) : C3) : Q8) . C3\\{} 
			[ 648, 532 ] & 3 & (((C3 x C3) : C3) : Q8) : C3\\{} 
			[ 648, 533 ] & 3 & (((C3 x C3) : C3) : Q8) : C3\\{} 
			[ 648, 534 ] & 3 & ((C3 x C3) : Q8) : C9\\{} 
			[ 648, 641 ] & 3 & ((C3 x C3 x C3) : Q8) : C3\\{} 
			[ 648, 702 ] & 3 & C3 x (((C3 x C3) : Q8) : C3)\\{} 
			[ 648, 703 ] & 4 & (((C3 x C3 x C3) : (C2 x C2)) : C3) : C2\\{} 
			[ 648, 704 ] & 4 & (((C3 x C3 x C3) : (C2 x C2)) : C3) : C2\\{} 
			[ 648, 705 ] & 3 & (S3 x S3 x S3) : C3\\{} 
			[ 648, 706 ] & 3 & C2 x (((C3 x C3 x C3) : (C2 x C2)) : C3)\\{} 
			[ 672, 1049 ] & 3 & C4 x ((C2 x C2 x C2) : (C7 : C3))\\{} 
			[ 672, 1256 ] & 3 & C2 x C2 x ((C2 x C2 x C2) : (C7 : C3))\\{} 
			[ 672, 1257 ] & 3 & (C2 x C2 x C2 x C2 x C2) : (C7 : C3)
	\end{tabular}}
\end{table}

\subsection{Equations in finite semigroups}
For a semigroup $S$, the problems $\EQNSAT(S)$ and $\EQNID(S)$ both receive two expressions as input. The questions is whether the two expressions evaluate to the same element under some (resp.\ all) assignments.
For semigroups $R, S$ we say that $R$ \emph{divides} $S$ if $R$ is a quotient of a subsemigroup of $S$. 
The following lemmas are straightforward to prove using basic semigroup theory.

For the proofs, we need Green's relations $\cH$ and $\cJ$. For a definition, we refer to \cite[Appendix A]{rs09qtheory}. For a semigroup $S$ we write $S^1$ for $S$ with an identity adjoined if there is none.

\begin{lemma}\label{lem:maximalsgEQN}
If $G$ is a maximal subgroup of a finite semigroup $S$, then $\EQNSAT(G) \leq_{\mathrm{m}}^{\AC} \EQNSAT(S)$.
\end{lemma}

\begin{proof}
	Let $e \in G$ denote the identity of $G$. Clearly,  $G = eGe \leq eSe$ and $eSe$ is a submonoid of $S$ with identity $e$.
	The reduction simply replaces every variable $X$ by $eXe$ (and likewise for constants). Let $\tilde \alpha$ denote the equation we obtain from an input equation $\alpha$ this way. Now the question is whether $\tilde \alpha = e$ in $S$.
	Clearly, if $\alpha$ has a solution in $G$, the resulting equation $\tilde \alpha$ has a solution in $S$. On the other hand, if $\tilde \alpha$ has a solution in $S$, we obtain a solution of $\alpha = e$ in $S$ where every variable takes values in $eSe$.

	Assume we have $\sigma(X) = x \not\in G$ for a satisfying assignment $\sigma$ and some variable $X$ of $\alpha$. Since $\sigma(\alpha) = e$, we have that $e$ is in the two-sided ideal $S^1xS^1$ generated by $x=exe$. By point 2. of \cite[Exercise A.2.2]{rs09qtheory} it follows that $x\in H_e = G$ where $H_e$ denotes the $\cH$-class of $e$ under Green's relations (for a definition, we refer to \cite{rs09qtheory}) and $G$ agrees with $H_e$ because $G$ is a maximal subgroup.
\end{proof}

\begin{lemma}\label{dividesmaximal}
	If a group $G$ divides a semigroup $S$, then $G$ divides already one of the maximal subgroups (\ie regular $\cH$-classes) of $S$.
\end{lemma}

\begin{proof}
	Let $U \leq S$ a subsemigroup and $\phi: U \to G$ a surjective semigroup homomorphism. Pick some arbitrary element $s \in U$ and let $e= s^{\omega}$ be the idempotent generated by $s$. Clearly, we have $\phi(e) = 1$. Now, the subsemigroup  $eUe \leq U$ still maps surjectively onto $G$ under $\phi$: by assumption for every $g \in G$ there is some $u_g \in U$ with $\phi(u_g) = g$; hence, $g = 1 g 1 = \phi(e) \phi(u_g) \phi(e) \in \phi(eUe)$.
	
	If $eUe$ is not contained in a maximal subgroup, then by point 2. of \cite[Exercise A.2.2]{rs09qtheory}, there is some $t \in eUe$ which is not $\cJ$-equivalent to $e$. Now, we can repeat the above process starting with $t$. This will decrease the size of $U$, so it eventually terminates.		
\end{proof}

\begin{corollary}\label{cor:semigroup} 
	Let $S$ be a finite semigroup and $G$ a group dividing $S$. If $\FitL(G) \geq 4 $ or $\FitL(G) =3 $ and $G/\cU_2G$ is not a 2-group, then $\EQNSAT(S)$ is not in \P under ETH.
\end{corollary}
\begin{proof}
	If $G$ with $\FitL(G) \geq 4 $ or $\FitL(G) =3 $ and $G/\cU_2G$ divides $S$, then it follows from \cref{dividesmaximal} that there is a group $\wt G$ with the same properties and which is a maximal subgroup of $S$. Hence, the statement 
	 follows from \cref{lem:maximalsgEQN}.	
\end{proof}

\cite[Theorem 1]{AlmeidaVG09} states that identity checking over $\wt G$ reduces to identity checking over $S$ where $\wt G$ is the direct product of all maximal subgroups of $S$. However, be aware that in this context the identity checking problem does not allow constants. Since the proof of \cref{thm:main} essentially relies on the fact that the subgroup $K$ is inducible and this can be only shown using constants, this does not allow us to show hardness of  $\EQNID(S)$.

\section{Conclusion}
We have shown that assuming the exponential time hypothesis there are solvable groups with equation satisfiability problem not decidable in polynomial time. Thus, under standard assumptions from complexity theory this means a negative answer to \cite[Problem 1]{BurrisL04} (also conjectured in \cite{Horvath11}). \cref{thm:main2} yields a quasipolynomial time lower bound under ETH. Thus, a natural weakening of \cite[Problem 1]{BurrisL04} is as follows:
\begin{conjecture}\label{conj:quasipoly}
	If $G$ is a finite solvable group, then $\EQNSAT(G)$ and $\EQNID(G)$ are decidable in quasipolynomial time.
\end{conjecture}

In \cite[Theorem 2]{BarringtonMMTT00} it is proved that $\ProgSAT(G) $ and, hence, also $ \EQNSAT(G)$ can be decided
in quasipolynomial time given that $G$ is \AND-weak. 
As remarked in \cref{sec:programs} this theorem remains valid with our slightly less restrictive definition of \AND-weakness in \cref{conj:andweak}. Thus, \cref{conj:andweak} implies \cref{conj:quasipoly}. 
In particular, under the assumption of both ETH and the \AND-weakness conjecture (\cref{conj:andweak}), for every finite solvable group $G$ meeting the requirements of \cref{thm:main2} there are quasipolynomial upper and lower bounds for $\EQNSAT(G)$ and $\EQNID(G)$~-- so under these assumptions both problems are neither in \P nor \NP-complete. This contrasts the situation for solving systems of equations: there is a clear \P versus \NP-complete dichotomy \cite{GoldmannR02}.

\cref{thm:main2} proves lower bounds on \EQNSAT and \EQNID for all sufficiently complicated finite solvable groups. Together with the authors of \cite{IdziakKK20} we can extend this to \emph{all} groups of Fitting length three \cite{IdziakKKW20arxiv}. 

 Possible further research might address the complexity of \EQNSAT and \EQNID in groups of Fitting length two. 
 Another direction for future work is the complexity of \EQNID for expressions without constants.

\bibliography{equations}

\end{document}